%% file: paper.tex
\documentclass[11pt,a4paper]{article}

\usepackage[T1]{fontenc}
\usepackage[utf8]{inputenc}
\usepackage[margin=1in]{geometry}
\usepackage[dvipsnames]{xcolor}
\usepackage{cite}
\usepackage{array}
\usepackage{booktabs}
\usepackage{multirow}
\usepackage{graphicx}
\usepackage{xfrac}

\usepackage{amsfonts}
\usepackage{amssymb}
\usepackage{amsmath}
\usepackage{amsthm}
\usepackage{mathtools}
\usepackage{subcaption}
\usepackage{algorithm}
\usepackage{algorithmic}
\usepackage{enumerate}

\usepackage{tikz}
\usetikzlibrary{positioning, calc, fit, math, trees}
\usetikzlibrary{decorations.pathreplacing,calligraphy}
\usepackage{todonotes}
\usepackage[norefs,nocites,nomsgs]{refcheck}

\usepackage[hidelinks]{hyperref}
\usepackage[capitalize,noabbrev]{cleveref}

\input{macros}

\newtheorem{theorem}{Theorem}[section]
\newtheorem{lemma}[theorem]{Lemma}

\newtheorem{observation}[theorem]{Observation}
\theoremstyle{definition}
\newtheorem{definition}[theorem]{Definition}

\theoremstyle{remark}

\title{On Time-Memory Tradeoffs for Maximal Palindromes\\ with Wildcards and $k$-Mismatches}

\author{
  Amihood Amir\thanks{Department of Computer Science, Bar-Ilan University, Israel, and Georgia Tech, College of Computing, GA, USA. Email: \texttt{amihoodamir@gmail.com}. Partially supported by ISF grant 168/23 and BSF grant 2018-141.}
  \and
  Ayelet Butman\thanks{Department of Computer Science, Holon Institute of Technology, Israel. Email: \texttt{ayeletbutman@gmail.com}.}
  \and
  Michael Itzhaki\thanks{Department of Computer Science, Bar-Ilan University, Israel. Email: \texttt{michaelitzhaki@gmail.com}. Partially supported by ISF grant 168/23.}
  \and
  Dina Sokol\thanks{Department of Computer and Information Science, Brooklyn College, NY, USA. Email: \texttt{dinasokol@gmail.com}.}
}
\date{}
\begin{document}
\maketitle
\vspace{-2.5em}

\input{abstract}

\input{main}

\bibliographystyle{plainurl}
\bibliography{paper}

\end{document}

%% file: macros.tex
\newcommand{\bigO}{\mathcal{O}}

\let\originalleft\left
\let\originalright\right
\renewcommand{\left}{\mathopen{}\mathclose\bgroup\originalleft}
\renewcommand{\right}{\aftergroup\egroup\originalright}

\newcommand{\Oh}[1]{\bigO\left(#1\right)}
\newcommand{\Oht}[1]{\tilde{\bigO}\left(#1\right)}

\newcommand{\dd}{\mathinner{.\,.}}
\newcommand{\emptystring}{\varepsilon}

\newcommand{\LCE}{\operatorname{LCE}}

\newcommand{\cconv}{\operatorname{cconv}}
\newcommand{\kconv}{\operatorname{kconv}}
\newcommand{\mconv}{\operatorname{mconv}}
\newcommand{\matches}{\approx}

\newcommand{\Ham}{\operatorname{H}}

\newcommand{\rev}[1]{\overleftarrow{#1}}

\newcommand{\Pal}{\mathcal{P}}

\DeclareMathOperator*{\argmax}{argmax}

%% file: abstract.tex
\begin{abstract}
\noindent This paper addresses the problem of identifying palindromic factors in texts that include wildcards -- special characters that match all others. These symbols challenge many classical algorithms, as numerous combinatorial properties are not satisfied in their presence.
We apply existing wildcard-LCE techniques to obtain a continuous time-memory tradeoff, and present the first non-trivial linear-space algorithm for computing all maximal palindromes with wildcards, improving the best known time-memory product in certain parameter ranges. Our main results are algorithms to find and approximate all maximal palindromes in a given text. We also generalize both methods to the $k$-mismatches setting, with or without wildcards.
\end{abstract}

%% file: main.tex
\section{Introduction}\label{s:intro}

\input{introduction/introduction}

\section{Preliminaries}\label{s:pre}

\input{preliminaries/preliminaries}

\section{Our Results in Detail}\label{sec:detailed-results}

\input{introduction/our_results_detailed}

\section{Maximal Palindrome Approximation and Computation}

\input{results/convolution}

\section{Palindromes with Mismatches and Wildcards}

\input{results/with_mismatches}

\section{Open Problems}

\input{future_work}

%% file: introduction/introduction.tex
Palindrome recognition is a fundamental problem in computer science. It appears in complexity courses as an example of a problem that is solved in linear time by a two-tape Turing machine~\cite{S:73} but requires quadratic time in a single-tape machine~\cite{M:84}. Apostolico et al.~\cite{ABG:92} considered parallel algorithms for the problem. Manacher~\cite{man:75} and Galil~\cite{G:75} showed how to use pushdown automata for recognizing palindrome prefixes of a string in the online setting. In the online setting, Amir and Porat~\cite{AP:14} showed how to recognize  {\em approximate} palindrome prefixes of a string. For more references, see~\cite{KMP:77,GS:76,LWY:19}. Many other palindrome-related problems have been further explored. For example, finding the palindromic length~\cite{BKRS:17}, indexing all palindromes~\cite{RS:18}, and sublinear time detection  algorithms~\cite{CPR:22}.

Beyond their theoretical interest, palindromic structures also appear in nature. Particular examples from biology include sequences recognized and cut by restriction enzymes, and the hairpin structures of DNA base pairs. For elaboration, see, e.g.~\cite{gk:97,f:03,lssnz:05,sr:12}.

{\em Wildcard Matching} is a variant of traditional string matching where certain input characters---referred to as {\em wildcards}---match all other characters. This problem has been extensively researched since the inception of modern pattern matching~\cite{FP:74,Abr:87} and continues to be an active field of research~\cite{I97,K:02,AI:24,BCS:24a,BCS:24b,cgl-04}.

A related problem is {\em mismatch matching} -- deciding whether a certain text matches a pattern up to at most $k$ mismatches. A common mismatch-counting scheme is the Hamming distance, which counts the number of mismatching characters between two strings of equal length. Another common and extensively researched scheme is edit-distance, which accounts for mismatches, as well as insertions and deletions. The mismatch matching problem has also drawn significant attention from theoretical computer scientists. See, e.g.~\cite{GG-86,AALLL:00,kps:06,fgku:15}.

In this work, we address the problem of {\em palindrome recognition with wildcards and mismatches}, which is related to all of the aforementioned areas of study. Under these settings, the challenge of finding palindromes in a given string becomes significantly harder, as most combinatorial properties used for existing efficient algorithms fail.

\subsection{Related Work}

\input{introduction/related_work}

\subsection{Our contributions.}

\input{introduction/our_contribution}

%% file: introduction/related_work.tex
\paragraph*{Wildcard Matching.} In wildcard matching, both the text $T$ and the pattern $P$ may contain wildcards. The first efficient algorithm for wildcard pattern matching was introduced in 1974 by Fischer and Paterson~\cite{FP:74}. Their algorithm employs a Fast Fourier Transform-based approach, and its time complexity is $\Oh{n\log m \log|\Sigma|}$, where $n=|T|,\,m=|P|$ and $\Sigma$ is the alphabet of $T$ and $P$. This result was later improved to a randomized $\Oh{n\log m}$ algorithm~\cite{I97,K:02}, followed by a deterministic $\Oh{n \log m}$ algorithm~\cite{CH:02,cc:07}. Fischer and Paterson's original method also computes the Hamming distance between two strings for every alignment of the pattern over the text in $\Oh{n|\Sigma|\log m}$ time. Building on these developments, Abrahamson~\cite{Abr:87} proposed a more efficient $\Oh{n\sqrt{m}\log m}$ algorithm for the same problem and introduced the concept of degenerate string matching, where each text symbol can represent a subset of the alphabet.

Beyond the baseline problem, additional tools were developed. Notable extensions include the dictionary index with wildcards and mismatches by Cole et al.~\cite{cgl-04} (often called ``don't care suffix tree'') and the {\em Longest Common Extension (LCE)} with wildcards time/memory tradeoff by Bathie et al.~\cite{BCS:24b}.

\paragraph*{Mismatch matching.}
A central notion in approximate string matching is \emph{pattern matching with mismatches}, where one computes (or reports) the mismatch distance (Hamming distance) between a pattern and a text. Its thresholded variant (\emph{$k$-mismatch}) outputs only alignments with at most $k$ mismatches. A standard toolbox reduces mismatch counting to convolutions/FFTs, yielding near-linear-time algorithms for many settings. For the bounded $k$-mismatch problem, early algorithms~\cite{LV-85} use LCE (``kangaroo jumps'') to give an overall $\Oh{nk}$ baseline ($n$ is the input length); a breakthrough by Amir et al.~\cite{alp:04} improved this to about $\Oh{n\sqrt{k\log k}}$, followed by a recent Monte Carlo $\Oh{n\sqrt k}$ algorithm by Chan et al.~\cite{CJWX:23}. Although an LCE structure exists for wildcard matching, there is little hope of finding a combinatorial LCE structure for the $k$-mismatch problem~\cite{BFGK:24}.

The $k$-mismatch problem and the wildcard matching problem are closely related and share a common algorithmic toolbox.

\paragraph*{Palindromes.} Despite being extensively researched, most existing works about palindromes only address classical settings, while generalizations (e.g., wildcards and mismatches) remain mostly undeveloped. Most results regarding palindrome processing are obtained through careful combinatorial analysis (e.g.~\cite{man:75,kk:09,BKRS:17,ItzhakiFCT:25,ItzhakiSOFSEM:26}), and these combinatorial properties do not readily generalize to our settings.

Nonetheless, existing works also investigate palindromes in the presence of wildcards and mismatches. Sokol~\cite{SOKOL:20} showed an algorithm for 2D palindromes with $k$ mismatches, using the Hamming distance measure. Alzamal et al.~\cite{MCCZSDS:23} considered palindromes in degenerate strings with gaps and mismatches. Amir et al.~\cite{AI:24} discussed the general problem of finding maximal palindromes in degenerate strings and proved lower bounds for general alphabets.

%% file: introduction/our_contribution.tex
We study the computation of all maximal palindromes in strings with wildcards, where each wildcard matches any alphabet symbol. As a baseline, we adapt known LCE-based techniques for palindromic computation to the wildcard setting, obtaining a new continuous time-memory tradeoff. We begin in the setting of no mismatches and introduce two linear-space algorithms:
\begin{enumerate}[i]
    \item \textbf{precise algorithm} that detects all maximal subpalindromes, and
    \item \textbf{approximation algorithm} that approximates all maximal subpalindromes.
\end{enumerate}

In regimes with high wildcard density (i.e., $G=\omega(\sqrt {n\log n})$), the first algorithm achieves a superior time-memory product compared to the baseline methods. The approximation algorithm achieves $\Oht{(1+1/\epsilon)n}$ time.

We then show that both algorithms generalize to the $k$-mismatch setting. There, the first algorithm is particularly effective as it can outperform the baseline algorithms in both time and space. However, the approximation algorithm no longer guarantees $\Oht{(1+1/\epsilon)n}$ running time; its complexity depends on the counting-convolution cost $f_S(n)$ (see~\cref{tab:results}).

A detailed statement of our guarantees, which relies on technical definitions, is presented in \cref{sec:detailed-results} after the preliminaries.

%% file: preliminaries/preliminaries.tex
\paragraph*{Basic definitions.} A {\em string} is an ordered sequence of characters. We denote the length of $S$ as $|S|$, its $i$-th character by $S[i]$ and its factor $S[i]S[i+1]\dots S[j]$ by $S[i\dd j]$. The empty string of length $0$ is denoted by $\emptystring$. When $j<i$, we denote $S[i\dd j]$ as the empty string $\emptystring$. A factor $S[1\dd i]$ is called a {\em prefix}, and a factor $S[j\dd |S|]$ is called a {\em suffix}. If $1\le i<|S|$ (resp. $|S|\ge j>1$) the prefix (resp. suffix) is said to be {\em proper}. The {\em concatenation} of strings $S$ and $T$ is $S\cdot T\coloneq S[1]S[2]\dots S[|S|]T[1]T[2]\dots T[|T|]$. When clear from context, we write $ST$. For a positive integer $i$, we denote $S^i$ as the concatenation of $S$ to itself $i$ times, i.e., $S^1 = S$, and $S^i = S^{i-1} \cdot S$. The {\em reverse} of $S$ is $\rev{S}=S[n]S[n-1]\dots S[1]$. A string satisfying $S=\rev{S}$ is called a {\em palindrome}.

To refer to a specific occurrence of a factor in a string, we use the interval notation $[i\dd j]=\{t\mid t\in \mathbb N,\; i \le t \le j\}$. An interval $[i\dd t]$ is called a {\em prefix} of $[i\dd j]$ if $t\le j$, and a suffix is defined symmetrically.
The center of an interval $\mathcal P=[i\dd j]$ is denoted as $c_\Pal \coloneq \frac{i+j}{2}$. The interval notation is used when a specific factor is concerned, rather than a string\footnote{Certain properties are only meaningful for substrings, such as starting and ending position. When palindromes are concerned, maximality is only well defined for substrings.}.

For two strings of equal length $S,\,T$, we denote by $\Ham(S,\,T)$ the {\em Hamming distance} between $S$ and $T$, i.e., the number of mismatching characters. Since comparing a string $S$ to its reverse $\rev{S}$ counts each mismatched position twice (once in the left half and once in the right half), we define the {\em palindromic mismatch count} of $S$ as $\frac{1}{2}\Ham(S,\,\rev{S})$. A string is called a {\em $k$-palindrome} if its palindromic mismatch count is at most $k$. For example, $S=\texttt{abcdcbx}$ has palindromic mismatch count $\frac{1}{2}\Ham(\texttt{abcdcbx},\,\texttt{xbcdcba})=1$, and indeed, replacing one character in $S$ suffices to make it a palindrome.

\paragraph*{Palindromes.}
A string $S$ is called a \textit{palindrome} or \textit{palindromic} if $S=\rev{S}$. For example, the strings $\texttt{abcba}$ and $\texttt{abccba}$ are palindromes, while $\texttt{abcbaa}$ is not. An interval $\Pal=[i\dd j]$ is called a {\em subpalindrome} of $S$ if its corresponding factor $S[i\dd j]$ is palindromic. We say that $\Pal$ is {\em centered} at $c_{\Pal}$, and that its {\em radius} is $\lceil \frac{j-i}{2}\rceil$. If the interval $[i\dd j]$ is a subpalindrome but the interval $[i-1\dd j+1]$ is either not defined or not a subpalindrome, we say that $[i\dd j]$ is a {\em maximal subpalindrome}. There are $2|S|-1$ center positions in $S$, each of which corresponds to a unique maximal subpalindrome. We denote the maximal subpalindrome at center $c$ as $\Pal_c$ and its radius as $r_c$.

To detect the radii of all maximal subpalindromes (from now on: all maximal palindromes), one can use Manacher's algorithm~\cite{man:75}, which runs in linear time and is considered to be efficient in practice. However, it is not suitable for mismatches and wildcards. An alternative is to detect palindromes using LCE queries, which can also be applied in wildcard/mismatch settings.

\begin{definition}[LCE Query]
    Let $S,\,T$ be two strings. The {\em Longest Common Extension (LCE)} query $\LCE(S,\,i,\,T,\,j)$ computes the largest integer $k \ge 0$ such that $S[i\dd i+k-1]=T[j\dd j+k-1]$.
\end{definition}

A linear-space LCE data structure can answer LCE queries in constant time, as summarized in the following lemma:
\begin{lemma}[LCE data structure~\cite{KLAAP:01}]\label{lem:lcp-ds}
Let $S,\, T$ be two strings over a linearly-sortable alphabet. A data structure that answers arbitrary LCE queries in constant time can be constructed in $\Oh{|S|+|T|}$ time, and the structure occupies $\Oh{|S|+|T|}$ space.
\end{lemma}

All palindromes of a string can be identified using LCE queries: $\LCE(S,\,i+1,\,\rev{S},\,n-i+2)$ equals $r_{i}$ and $\LCE(S,\,i+1,\, \rev{S},\,n-i+1)$ equals $r_{i+0.5}$.

A simple yet efficient technique, often referred to as ``kangaroo jumps'', utilizes LCE queries to skip over matching text segments in constant time. By repeatedly querying the LCE to identify the next mismatch, one can compute LCE with $k$ mismatches in $\Oh{k}$ time, and consequently, can find all maximal palindromes with $\le k$ mismatches in $\Oh{nk}$ time.

\paragraph*{Wildcard- and Mismatch-Matching.}

In this paper, we refer to strings containing wildcards as {\em wildcard strings}, and otherwise as solid strings. As this paper mainly concerns wildcard strings, an unspecified string is assumed to be a wildcard string. We denote a wildcard---a character that matches any character---as $\phi$. Given two strings $S,\, T$ of length $n$, we say that $S$ {\em matches} $T$ if $S[i]\matches T[i]$ for every $1\le i \le n$, and denote $S\matches T$. For example, $\texttt{ab}\phi \matches \phi \texttt{ba}$.

The problem of {\em pattern matching with wildcards}, i.e., finding all occurrences of a wildcard pattern of length $m$ within a wildcard text of length $n$, can be solved in $\Oh{n\log m}$ time by using fast polynomial multiplication algorithms--commonly the {\em Fast Fourier Transform (FFT)}. This method is also referred to as {\em convolution}. The convolution method also matches all prefixes of the pattern with suffixes of the text and vice versa, a property vital for palindrome detection. A small visualization of the convolution method is shown in~\cref{fig:conv}.

Abrahamson~\cite{Abr:87} showed how to compute the Hamming distance between every alignment of the pattern over the text, including all prefix/suffix alignments.

In this paper, we use both tools and refer to them as the {\em matching-} and {\em counting-convolution}. The formal definition follows:
\begin{definition}[Convolution Array]\label{def:conv-arr}
    Let \( T,\, P \) be two strings of respective lengths $n\ge m$. The {\em counting convolution} \( \operatorname{cconv}(T,\, P) \) is an array $C[1\dd n+m-1]$, where:
    \[
    C[i] \coloneq
    \begin{cases}
    \Ham(T[1\dd i],\,P[m-i+1\dd m]) & \text{for } 1 \leq i < m \\
    \Ham(T[i-m+1\dd i],\,P) & \text{for } m \le i \le n\\
    \Ham(T[i-m+1 \dd n],\,P[1\dd n-i+m]) & \text{for } n < i \le n+m-1
    \end{cases}
    \]

    Similarly, the {\em matching convolution} $\mconv(T,\,P)$ is a binary array $M[1\dd n+m-1]$ where:
    \[
    M[i]\coloneq \begin{cases}
        1 & \text{if } \cconv(T,\,P)[i] =0 \\
        0 & \text{otherwise}
    \end{cases}
    \]
\end{definition}

The array $\mconv(T,\,P)$ can be computed in $\Oh{n\log m}$ time and the arrays $\cconv(T,\,P)$ can be computed in $\Oh{n\min(|\Sigma|\log m,\sqrt{m\log m})}$ time. We emphasize that $\mconv$ is not computed from $\cconv$; we define it in terms of $\cconv$ for clarity. An example of the counting and matching convolution arrays is in~\cref{fig:conv}.
\input{figures/conv}

Next, we define wildcard fragments:
\begin{definition}[Wildcard Fragments]
    We say that the factor $\Pal=[i\dd j]$ is a {\em wildcard fragment} of string $S$ if it satisfies all of the following:
    \begin{enumerate}[i]
        \item For every $i\le k\le j$, the character $S[k]$ is a wildcard.
        \item $i-1\le 0$ or $S[i-1]$ is not a wildcard.
        \item $j+1>|S|$ or $S[j+1]$ is not a wildcard.
    \end{enumerate}
\end{definition}

A parameter often used by algorithms for wildcard strings is the number of different wildcard fragments (or simply fragments) in the input, which we call $G$. For example, the string $S=\phi\texttt{aab}\phi\phi\phi\texttt{a}\phi\texttt{a}$ has three fragments: $S[1\dd1],\, S[5\dd 7],\, S[9\dd 9]$. The number of fragments ranges from zero (a solid string) to $\Theta(n)$ (e.g., the string $(a\phi)^n$).

The baseline algorithm for this paper uses LCE with wildcards (LCEW):
\begin{lemma}[LCEW data structure~\cite{BCS:24b}]
    Let $S$ be a string of length $n$ with $G$ wildcard fragments. For any parameter $T \in [1, G]$, one can build an LCEW data structure in $\Oh{\frac{nG}{T}\log n}$ time using $\Oh{\frac{nG}{T}}$ space. This structure answers LCEW queries in $\Oh{T}$ time.
\end{lemma}

By selecting the tradeoff parameter to be $T\coloneq \sqrt{G}$, we obtain an $\Oht{n\sqrt{G}}$-time and space algorithm to detect all maximal palindromes in $S$.

%% file: figures/conv.tex
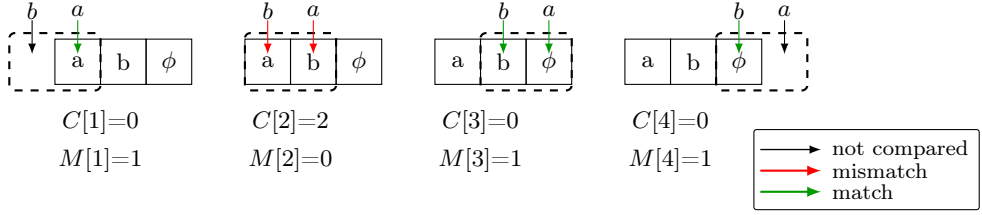
\begin{figure}[ht]
\centering
\begin{tikzpicture}[>=latex, node distance=0, every node/.style={font=\footnotesize}]
  \tikzset{
    box/.style={draw, minimum width=1.2cm, minimum height=0.9cm},
    gbox/.style={draw, dashed, color=gray, minimum width=1.2cm, minimum height=0.9cm},
    win/.style={rounded corners=2pt, dashed, draw, thick, inner sep=0pt},
  }
  \def\cell{0.6}
  \def\gap{0.9}
  \def\stepw{(3*\cell + \gap)}

\newcommand{\sliding}[6]{
    \begin{scope}[xshift=#1]

      \foreach \i/\name in {0/a,1/b,2/$\phi$}{
        \draw[box] (\i*\cell,0) rectangle ++(\cell,\cell) node[pos=.5]{\name};
      }

      \pgfmathsetmacro\leftx{(#2)*\cell}
      \draw[win] (\leftx,-0.08) rectangle ++(2*\cell,\cell+0.16);

      \draw[->, #5] (\leftx+0.5*\cell,\cell+0.25) -- ++(0,-0.45);
      \node at (\leftx+0.5*\cell,\cell+0.38) {$b$};

      \draw[->, #6] (\leftx+1.5*\cell,\cell+0.25) -- ++(0,-0.45);
      \node at (\leftx+1.5*\cell,\cell+0.38) {$a$};

      \pgfmathsetmacro\idx{int((#2)+2)}
      \node at (1*\cell,-0.5) {$C[\idx]$=#3};
      \node at (1*\cell,-1) {$M[\idx]$=#4};
    \end{scope}
}

\sliding{-2.5cm}{-1}{0}{1}{black}{green!60!black}
\sliding{0cm}{0}{2}{0}{red}{red}
\sliding{2.5cm}{1}{0}{1}{green!60!black}{green!60!black}
\sliding{5cm}{2}{0}{1}{green!60!black}{black}

  \node[anchor=south, draw, fill=white, rounded corners=1pt, inner sep=2pt]
    at ([xshift=2em,yshift=-1em]current bounding box.south east) {
      \begin{tikzpicture}[baseline=0]
        \draw[line width=0.5pt, ->] (0,0) -- +(0.8,0);
        \node[anchor=west, xshift=2pt, font=\scriptsize] at (0.8,0) {not compared};

        \draw[red, line width=0.8pt, ->] (0,-0.28) -- +(0.8,0);
        \node[anchor=west, xshift=2pt, font=\scriptsize] at (0.8,-0.28) {mismatch};

        \draw[green!60!black, line width=0.8pt, ->] (0,-0.56) -- +(0.8,0);
        \node[anchor=west, xshift=2pt, font=\scriptsize] at (0.8,-0.56) {match};
      \end{tikzpicture}
    };
\end{tikzpicture}

\caption{The counting and matching convolution between $T=\texttt{ab}\phi$ and $P=\texttt{ba}$. We denote $C=\cconv(T,\,P)$ and $M=\mconv(T,\,P)$.}
\label{fig:conv}
\end{figure}

%% file: introduction/our_results_detailed.tex
\begin{table}[ht]
\centering
\caption{Asymptotic space and time bounds for computing all maximal $k$-mismatch palindromes with wildcards. 
Parameters: $n$ (input length), $k$ (maximum mismatches allowed), $G$ (fragments in the input), and $\Sigma$ (the alphabet). We denote by $f_S(m)\coloneq \Oh{m\min(|\Sigma|\log m, \sqrt{m\log m})}$ the time required to compute counting convolution. Our results are highlighted in bold and with a star.
}
\label{tab:results}
\renewcommand{\arraystretch}{1.2}
\begin{tabular}{@{}llcc@{}}
\toprule
\textbf{Algorithm} & \textbf{Regime} & \textbf{Space} & \textbf{Time} \\ \midrule
\multirow{3}{*}{LCE-based}
& $k=0,\,G=\Omega(\log n)$ & $\Oh{n\sqrt{G/\log n}}$ & $\Oh{n\sqrt{G\log n}}$ \\
& $k>0,\; G =\Omega(k\log n)$ & $\Oh{n\sqrt{Gk/\log n}}$ & $\Oh{n\sqrt{Gk\log n}}$ \\
& otherwise & $\Oh{n}$ & $\Oh{n(k+G)}$ \\
\midrule
\multirow{2}{*}{\textbf{Precise${}^\star$}}
& $k=0$ & $\Oh{n}$ & $\Oh{n\sqrt{n\log n}}$ \\
& otherwise & $\Oh{n}$ & $\Oh{n\sqrt{f_S(n)}}$ \\
\midrule
\multirow{3}{*}{\textbf{\((1+\epsilon)\)-Approx.${}^\star$}}
& $k=0$ & $\Oh{n}$ & $\Oh{(1/\epsilon)\,n\log^2 n}$ \\
& $k>0,\; G=0$ & $\Oh{n}$ & $\Oht{(1/\epsilon)nk^{2/3}}$ \\
& otherwise & $\Oh{n}$ & $\Oh{(1/\epsilon)\log n \cdot f_S(n)}$ \\
\bottomrule
\end{tabular}
\end{table}

The guaranteed running time of the algorithms in this paper is presented in~\cref{tab:results}.

\begin{itemize}
\item \textbf{For $k=0$ (Precise, no mismatches):} The LCE algorithm is faster in raw time whenever $G=o(n)$. However, its space can grow to $\Theta\!\left({n\sqrt{n/\log n}}\right)$. The precise algorithm provides the first non-trivial linear-space solution, and its $\Oh{n^{2.5}\sqrt{\log n}}$ time/memory product is superior to the LCE method for $G = \omega(\sqrt{n\log n})$.

\item \textbf{For $k>0$:} The precise algorithm gives the best time guarantee: $\Oh{n^{7/4}\log^{1/4} n}$.

\item \textbf{Approximation:} From a runtime perspective, the approximate algorithm is the fastest in all regimes, unless the parameters $k$ and $G$ are very small (e.g., $k+G=o(\log^2 n)$).
\end{itemize}

\paragraph*{Remark on the approximation algorithm.}
In the $(1+\epsilon)$-approximation algorithm with mismatches, the running time improves by a factor of $\Oh{\log n}$ when $|\Sigma|=\Omega(n^\alpha)$ for some constant $\alpha>0$. The reason is that the algorithm computes counting convolutions on strings whose lengths grow geometrically, so its total convolution cost is of the form $\sum_i 2^i f_S(n/2^i)$. Under this alphabet-size assumption, this sum is geometric and therefore evaluates to $\Oh{f_S(n)}$.

%% file: results/convolution.tex
The main contribution of this paper is the development of algorithms to find and approximate all maximal palindromes in a wildcard string.

We begin by sketching the LCE baseline algorithm for wildcards. The idea is similar to kangaroo jumping: we treat wildcards as regular characters, which causes the LCE query to stop early when comparing wildcards to other characters. Therefore, if there are $g$ wildcards, this method requires $\Oh{g}$ time per LCE query. A better way is to jump to the end of the wildcard fragment when such a mismatch is detected: wildcard fragments are guaranteed to match anything. This improvement reduces the LCE query time from $\Oh{g}$ to $\Oh{G}$. 

\subsection{Setup}
Throughout the section, we denote by $S$ the input wildcard string of length $n$. Recall that $\Pal_c$ is the maximal subpalindrome at center $c$, and that $r_c$ is the radius of $\Pal_c$.

In this section, we compute the convolution of $S$ with its reverse, i.e., $\operatorname{cconv}(S,\;\rev{S})$ or $\mconv(S,\;\rev{S})$. We call these {\em self-convolutions}, and denote them by $\cconv(S)$ and $\mconv(S)$, respectively. We observe that the self-convolution $\mconv(S)$ identifies all prefix and suffix palindromes in $S$. Later in the algorithms, we use the self-matching-convolution array as if it reported palindromes.
\begin{observation}[Palindromes detection using convolutions]\label{ob:pref-suf}
    Let $S$ be a wildcard string of length $n$, and let $M=\mconv(S)$ be the self-matching-convolution of $S$. The string $S$ has a palindromic prefix $S[1\dd i]$ if and only if $M[i]=1$, and the string $S$ has a palindromic suffix $S[i\dd n]$ if and only if $M[n+i-1]=1$.
\end{observation}

\begin{proof}
A prefix $S[1\dd i]$ is a palindrome if it matches its reverse, $\rev{S[1\dd i]}$, which is equivalent to the suffix $\rev{S}[n-i+1 \dd n]$. The matching convolution index $M[i]$ is defined to compute the match between this specific prefix-suffix alignment. The suffix case $M[n+i-1]$ is symmetrical.
\end{proof}
Similarly, the self-counting-convolution $\cconv(S)$ counts for every prefix and suffix of $S$ the number of mismatches with its reverse.
More generally, when the self-matching-convolution is computed on a factor $S[\ell\dd r]$, it identifies all centers $\ell \le c \le r$ for which $S[\ell\dd r]$ has a palindromic prefix or suffix at center $c$; the detected radius at $c$ is $\lceil \min\{ c-\ell,\;r-c \} \rceil$.

\textbf{Simplifying assumption.} Since centers of even-length palindromes are half-integers, they complicate indexing notations. Therefore, we ignore even-length palindromes in the input and hence define the center of the interval $\Pal=[i\dd j]$ to be $c_\Pal\coloneq \lfloor\frac{i+j}{2}\rfloor $. The assumption is valid because the transformation $S'=S[1]\$S[2]\$\dots\$S[n]$ makes any subpalindrome in $S$ into an odd-length subpalindrome in $S'$.

We proceed to detail two naive subroutines used by our algorithms:
\begin{itemize}
    \item The procedure $\operatorname{NaivePalFind}(S,\,c,\,u)$ attempts to compute $r_c$ by performing character-by-character comparison, but no more than $u$ comparisons. The procedure is guaranteed to return $\min(u,\,r_c)$ and run in time $\Oh{u}$.
    \item The procedure $\operatorname{NaivePalExtend}(S,\,c,\,r_0)$ is given a center $c$ and a guaranteed radius $r_0$. It naively extends this palindrome by comparing the characters at index $c \pm (r_0+i)$ for $i=1,2,\dots$, stopping at the first mismatch.
\end{itemize}
The procedures are illustrated using a toy example in~\cref{fig:naive}.

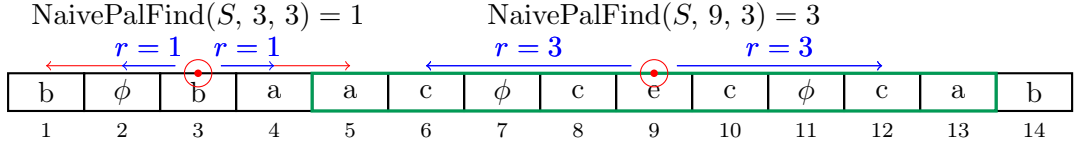
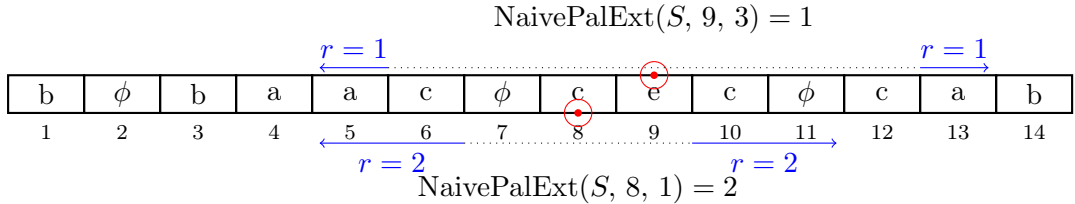
\begin{figure}
    \centering
\begin{subfigure}{\textwidth}
    \input{figures/naive_methods}
    \caption{Illustration of $\operatorname{NaivePalFind}$. The blue arrow indicates successful comparisons performed by the algorithm, and the red arrow indicates a mismatch. The maximal palindrome centered at $c=9$ is highlighted in green.}
\end{subfigure}

\vspace{0.5em}
\hrule
\vspace{0.5em}

\begin{subfigure}{\textwidth}
    \input{figures/naive_ext}
    \caption{Illustration of $\operatorname{NaivePalExt}$. The dashed line indicates the highest known lower bound on the radius of the palindrome at the given center (the third parameter). The blue arrow indicates successful comparisons performed by the algorithm.}
\end{subfigure}

    \caption{Demonstration of the naive subroutines $\operatorname{NaivePalFind}$ and $\operatorname{NaivePalExtend}$. The red dots are the centers provided as input to the respective subroutines.}
    \label{fig:naive}
\end{figure}

\subsection{Precise Algorithm}

\input{results/convolution_precise}

\subsubsection{Example: Precise Algorithm}\label{ssec:toy-run}
Consider the string $S=\texttt{b}\phi\texttt{baac}\phi\texttt{cec}\phi\texttt{cc}$ of length $n=13$, with block size $u=3$. Phase~1 of~\cref{alg:precise} iterates over prefixes $S[1\dd iu]$ for $i=1,2,\dots$, computing the self-matching-convolution of each prefix to detect palindromic suffixes. \cref{fig:precise} depicts five such iterations. A center marked in green is one whose palindrome extends to (or past) the current prefix boundary, so that the convolution detects it as a palindromic suffix. Once a green center turns red in a later iteration, its palindrome no longer reaches the boundary of the longer prefix -- meaning Phase~1 has already captured its radius to within an additive error of $u$. At that point, the center is handed off to Phase~2, where $\operatorname{NaivePalExtend}$ closes the remaining gap of at most $u$ characters.

For instance, center $c=8$ first appears as green in iteration $i=3$ (prefix $S[1\dd 9]$) and turns red in iteration $i=4$ (prefix $S[1\dd 12]$), indicating that the palindrome at $c=8$ does not extend to position $12$. The lower bound $\tilde{R}[8]$ recorded during iteration $i=3$ is then refined by naive extension.

\input{figures/example_precise}

\subsection{Approximation Algorithm}

\input{results/convolution_approx}

\subsubsection{Example: Approximation Algorithm}\label{ssec:example-approx}
Consider the string $S=\texttt{b}\phi\texttt{baac}\phi\texttt{cec}\phi\texttt{ccab}\phi$ of length $n=16$. \cref{fig:approx} traces the behavior of~\cref{alg:approx} for center $c=5$ across five iterations. The red dots mark the anchor positions, and the blue brace indicates the window $[a-u,\;a+u]$ around the anchor closest to $c=5$. The gray region highlights the substring of the window that determines the detected radius at $c=5$: its extent is $\min(c-\ell,\;r-c)$, where $[\ell, r]$ is the window boundary.

Observe the two mechanisms at work. First, the window size $u$ grows geometrically (from $u=2$ to $u=8$ across the five panels), so the detected radius at $c=5$ increases with each iteration. Second, the anchor set is thinned: the first two panels have eight anchors, the next two have four, and the last panel has two. This thinning occurs whenever $u$ doubles, keeping the per-iteration cost proportional to $n\log n$. Despite the coarser anchor spacing, the growing window size ensures that each center remains within distance $u/2$ of some anchor, preserving the $(1+\epsilon)$ approximation guarantee.

\input{figures/example_approx}

%% file: figures/naive_methods.tex
\newcommand{\naivepal}[3]{
  \node at ({#1 + 0.5}, 1.25) {$\operatorname{NaivePalFind}(S,\,#1,\,#2)=#3$};

  \node[circle, draw=red] at ({#1 + 0.5}, 0.5) {};
  \node[circle, fill=red, inner sep=1pt] at ({#1 + 0.5}, 0.5) {};
    
    \ifthenelse{#2 = #3} {
\draw[blue, ->] ({#1 + 0.2}, 0.6) -- node[midway, above] {$r=#3$} ({#1 - #3 + 0.5}, 0.6);
\draw[blue, ->] ({#1 + 0.8}, 0.6) -- node[midway, above] {$r=#3$} ({#1 + #3 + 0.5}, 0.6);
    }{
\draw[blue, tips=never] ({#1 + 0.2}, 0.6) -- node[midway, above] {$r=#3$} ({#1 - #3 + 0.5}, 0.6);
\draw[blue, tips=never] ({#1 + 0.8}, 0.6) -- node[midway, above] {$r=#3$} ({#1 + #3 + 0.5}, 0.6);
\draw[red, ->] ({#1 - #3 + 0.5}, 0.6) -- ({#1 - #3 - 0.5}, 0.6);
\draw[red, ->] ({#1 + #3 + 0.5}, 0.6) -- ({#1 + #3 + 1.5}, 0.6);
    }

    \draw[blue, ->] ({#1 + 0.2}, 0.6) -- node[midway, above] {$r=#3$} ({#1 - #3 + 0.5}, 0.6);
    \draw[blue, ->] ({#1 + 0.8}, 0.6) -- node[midway, above] {$r=#3$} ({#1 + #3 + 0.5}, 0.6);
}

\begin{tikzpicture}
  \foreach \v [count=\i] in {b, $\phi$, b, a, a, c, $\phi$, c, e, c, $\phi$, c, a, b} {
    \draw[thick] (\i,0) rectangle (\i+1,0.5);
    \node at (\i+0.5,0.25) {\v};
    \node at (\i+0.5,-0.25) {\scriptsize \i};
  }
    \draw[ForestGreen, very thick] (5,0) rectangle (14,0.5);

    \naivepal{9}{3}{3}
    \naivepal{3}{3}{1}

\end{tikzpicture}

%% file: figures/naive_ext.tex
\newcommand{\naivextex}[7]{
  \node at ({#1 + 0.5}, {#7}) {$\operatorname{NaivePalExt}(S,\,#1,\,#2)=#3$};

  \node[circle, draw=red] at ({#1 + 0.5}, #4) {};
  \node[circle, fill=red, inner sep=1pt] at ({#1 + 0.5}, #4) {};

\draw[dotted] ({#1 - #2}, #6) -- ({#1 + #2 + 1}, #6);
\draw[blue, ->] ({#1 - #2}, #6) -- node[midway, #5] {$r=#3$} ({#1 - #2 - #3 + 0.1}, #6);
\draw[blue, ->] ({#1 + #2 + 1}, #6) -- node[midway, #5] {$r=#3$} ({#1 + #2 + #3 + 0.9}, #6);
}

\newcommand{\naivext}[3]{\naivextex{#1}{#2}{#3}{0.5}{above}{0.6}{1.25}}
\newcommand{\naivextd}[3]{\naivextex{#1}{#2}{#3}{0}{below}{-0.4}{-1}}

\begin{tikzpicture}
  \foreach \v [count=\i] in {b, $\phi$, b, a, a, c, $\phi$, c, e, c, $\phi$, c, a, b} {
    \draw[thick] (\i,0) rectangle (\i+1,0.5);
    \node at (\i+0.5,0.25) {\v};
    \node at (\i+0.5,-0.25) {\scriptsize \i};
  }

    \naivext{9}{3}{1}
    \naivextd{8}{1}{2}

\end{tikzpicture}

%% file: results/convolution_precise.tex
To detect all maximal palindromes, our algorithm is parameterized by an integer $u$, which defines a block size. Every index of the form $i\cdot u$ is called an anchor. We denote $w\coloneq \lceil n/u\rceil$ as the number of blocks.

The algorithm consists of two main phases:

\begin{itemize}
\item \textbf{Phase 1: Coarse Radius Computation}
We compute a lower bound for the radius of every palindrome. We maintain a global array $\tilde{R}$ of lower bounds, initialized with zeros. We then perform $w$ iterations:
\begin{itemize}
    \item For each $i \in \{1,\dots,\,w\}$, we compute the self-convolution $M_i = \mconv(S[1 \dd \min\{iu,\,n\}])$.
    \item The resulting array $M_i$ identifies all the palindromic prefixes and suffixes of $S[1\dd iu]$. For every center $c$ that was identified as a palindrome with radius $r$, we set $\tilde{R}[c] \leftarrow \max\{\tilde{R}[c],\,r\}$.
\end{itemize}

\item \textbf{Phase 2: Naive Extension.}
After Phase 1, the array $\tilde{R}$ has all radii computed with a maximal additive error of $u$ per radius. Formally, the array $\tilde{R}$ guarantees for every center $c$ that $r_c-u< \tilde{R}[c]\le r_c$. To compute the radius precisely, we update $R[c]\gets \tilde{R}[c]+\operatorname{NaivePalExtend}(S,\,c,\,\tilde{R}[c])$ for every center position.
\end{itemize}
The full algorithm is presented in~\cref{alg:precise}, and an example can be found in~\cref{ssec:toy-run}.

\begin{lemma}[Correctness]
    When~\cref{alg:precise} finishes, $R$ contains the maximal radius $r_c$ for every center position $c$.
\end{lemma}
\begin{proof}
The algorithm is correct because $\tilde{R}$ never overestimates radii. If $\tilde{R}[c]=r$, it means that the self-matching-convolution detected a palindrome with radius $r$ at center $c$, hence $r_c\ge r$. Consequently, the naive extension in phase 2 begins with a valid lower bound $\tilde{R}[c]$ and extends it until a mismatch is found or no further extension is possible, which is the definition of maximal subpalindrome, as required.
\end{proof}

\begin{lemma}[Runtime]
    The runtime of~\cref{alg:precise} is $\Oh{\frac{n^2}{u}\log n+nu}$.
\end{lemma}
\begin{proof}
The algorithm is divided into two phases: the coarse radius computation (FFT) and the naive extension. The coarse radius computation runs $w$ times, each time performing a matching convolution in $\Oh{n\log n}$ time, hence the total time of phase 1 is $\Oh{wn\log n}=\Oh{\frac{n^2}{u}\log n}$.

The second phase performs naive extension for each center. We proceed to prove that for every center, the naive subroutine performs $\Oh{u}$ comparisons.

Assume for contradiction that for some center $c$, the routine $\operatorname{NaivePalExtend(S, c, \tilde{R}[c])}$ extended more than $u$ characters, i.e., $r_c-\tilde{R}[c]\ge u$.

Since the difference between the estimated radius $\tilde{R}[c]$ and the real radius $r_c$ is at least $u$, an anchor $ju$ exists such that $c+\tilde{R}[c] < ju \le c+r_c$.

In phase one, we compute the matching-self-convolution of all prefixes of the form $S[1\dd iu]$, including $\mconv(S[1\dd ju])$. Since $c+r_c\ge ju$, the string $S[1\dd ju]$ has a palindromic suffix at center $c$, and by \cref{ob:pref-suf} we would have detected it. However, since $c+\tilde{R}[c]< ju$, it follows that no palindromic suffix at center $c$ was identified, contradicting the observation.

Since $r_c-\tilde{R}[c]<u$, each of the $\Oh{n}$ invocations of $\operatorname{NaivePalExtend}$ performs $\Oh{u}$ comparisons and a total of $\Oh{nu}$ work, as required.
\end{proof}
The space consumption is linear because all the subroutines we invoke consume linear space, and the only non-constant memory we store is the arrays $R$ and $\tilde{R}$.

We conclude by balancing the runtime. The running time of~\cref{alg:precise} is $T(n)= \Oh{\frac{n^2}{u}\log n+nu}$ or $T(n)= \Oh{n\left(\frac{n}{u}\log n+u\right)}$. We minimize $T(u)$ by balancing $\frac{n}{u}\log n$ and $u$:
\[
\frac{n}{u} \log n \approx u \implies u^2 \approx n \log n \implies u = \sqrt{n \log n}.
\]
By plugging $u=\sqrt{n \log n}$ into $T(n)$, the overall running time is $T(n)= \Oh{n \sqrt{n \log n}}$, as summarized in the following theorem:
\begin{theorem}
    Let $S$ be a wildcard string of length $n$. The radii of all maximal subpalindromes of $S$ can be found in $\Oh{n\sqrt{n\log n}}$ time and $\Oh{n}$ space.
\end{theorem}

\begin{algorithm}[ht]
\caption{Detect All Maximal Subpalindromes}
\label{alg:precise}
\begin{algorithmic}[1]
    \REQUIRE String $S$ of length $n$, block size $u\ge 1$.
    \ENSURE An array $R$ satisfying $R[c]=r_c$ for every center $c$.
    \STATE $\tilde{R} \gets (0,\ldots,0)$ \COMMENT{Radius lower bound for each center.}
    \STATE $R \gets (0,\ldots,0)$ \COMMENT{Precise radius for each center.}
    \FOR{$i \gets 1$ \TO $\lceil \frac{n}{u}\rceil $} \label{alg:for-i}
        \STATE $M \gets \mconv\!\big(S[1\dd \min\{i\cdot u,\, n\}])$ \label{alg:conv}
        \FORALL{palindromes $(c,\,r)$ identified by $M$} \label{alg:matches}
            \STATE $\tilde{R}[c] \gets \max\{\tilde{R}[c],\, r\}$ \label{alg:update}
        \ENDFOR
    \ENDFOR
    \FORALL{centers $c$}
        \STATE $R[c] \gets \tilde{R}[c]+\operatorname{NaivePalExtend}(S,\,c,\,\tilde{R}[c])$
    \ENDFOR
    \RETURN $R$
\end{algorithmic}
\end{algorithm}

%% file: figures/example_precise.tex
\newcommand{\stringtemplate}[3]{
  \foreach \v [count=\i] in {b, $\phi$, b, a, a, c, $\phi$, c, e, c, $\phi$, c, c} {
    \draw[thick] (\i,0) rectangle (\i+1,0.5);
    \node (v\i) at (\i+0.5,0.25) {\v};
    \node at (\i+0.5,-0.25) {\scriptsize \i};
  }

    \foreach \v in {#2} {
      \node[circle, draw=ForestGreen] at ({\v + 0.5}, 0.5) {};
      \node[circle, fill=ForestGreen, inner sep=1pt] at ({\v + 0.5}, 0.5) {};
    }
    \foreach \v in {#3} {
      \node[circle, draw=red] at ({\v + 0.5}, 0.5) {};
      \node[circle, fill=red, inner sep=1pt] at ({\v + 0.5}, 0.5) {};
    }

  \draw[blue,thick,decorate,decoration={brace,mirror,raise=2em,amplitude=3pt},yshift=-2em]
  ($ (v1) + (-0.3, 0) $)
  --node[midway,below=6pt+1.5em] {$iu=#1$} 
  ($ (v#1) + (0.3, 0) $);

}

\newcommand{\newsubfigure}[2][1]{
    \begin{subfigure}{\textwidth}
    \centering
    \begin{tikzpicture}[scale=#1]
        #2
    \end{tikzpicture}
    \end{subfigure}
}

\newcommand{\subfigsep}{
    \vspace{-0.5em}
    \hrule
    \vspace{0.5em}
}

\begin{figure}[ht]
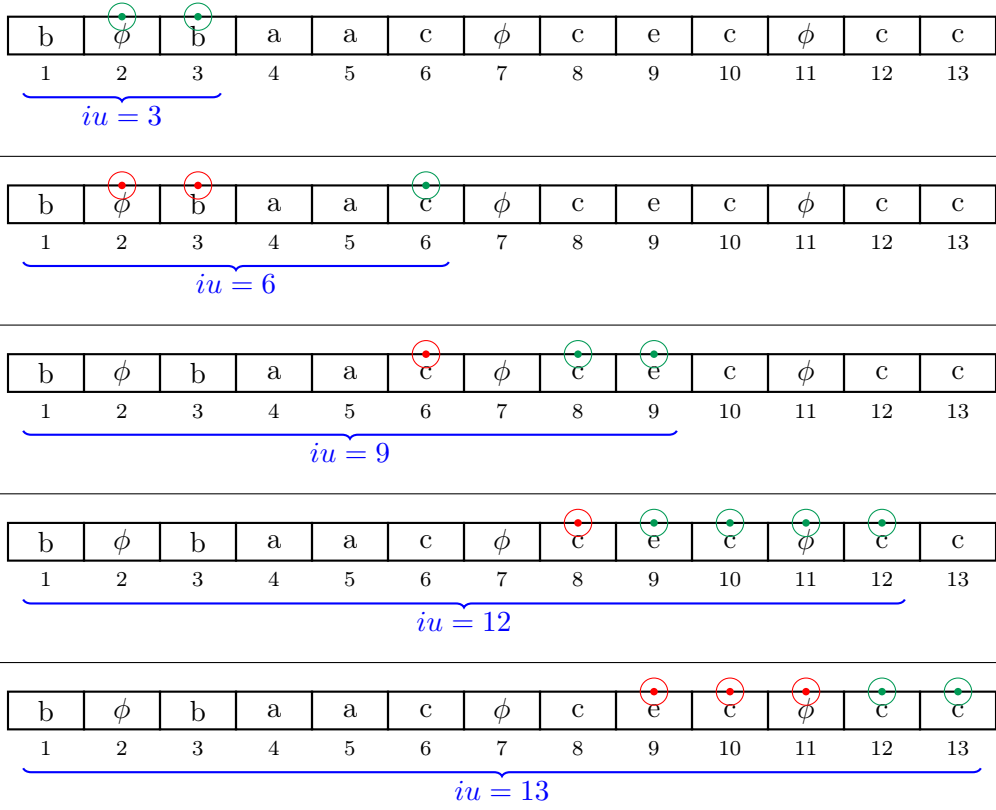

    \newsubfigure{
    \stringtemplate{3}{2,3}{}
    }
    \subfigsep
    \newsubfigure{
        \stringtemplate{6}{6}{2,3}
    }
    \subfigsep
    \newsubfigure{
        \stringtemplate{9}{9,8}{6}
    }
    \subfigsep
    \newsubfigure{
        \stringtemplate{12}{12,11,10,9}{8}
    }
    \subfigsep
    \newsubfigure{
        \stringtemplate{13}{13, 12}{9,10,11}
    }
    \caption{Demonstration of the precise algorithm. At every step, we search for the longest palindromic suffixes. A center identified as a palindrome is marked with green. Once a center is no longer identified, it is marked in red, and the naive extension takes place.}
    \label{fig:precise}
\end{figure}

%% file: results/convolution_approx.tex
\newcommand{\R}{\tilde{R}}

The approximation algorithm finds a $(1+\epsilon)$ approximation for all palindromes. Formally, given a string $S$ and an approximation parameter $\epsilon>0$, the approximation algorithm computes an array $\R$ such that for every center $c$ the array satisfies $\R[c]\le r_c\le (1+\epsilon)\tilde{R}[c]$. Throughout the section, we denote $\delta=\frac{1}{2}\epsilon$.

\paragraph*{Simplifying assumptions.} We require $\epsilon\le 0.5$. Additionally, we assume that $|S|$ is a natural power of $2$.

\paragraph*{High-level description.} The idea is to run the self-matching-convolution on exponentially-increasing factors of $S$. The factors are defined by a set of anchors $\mathcal A$ and a window size $u$, and are chosen to be $S[a-u\dd a+u]$ for every $a\in \mathcal A$. For example, if $\mathcal A=\{6, 8\}$ and $u=4$, the factors concerned are $S[2\dd 10]$ and $S[4\dd 12]$.

The algorithm initializes an approximation array $\R$, sets the initial window size to $u=u_0\approx \frac{1}{\delta}$, and chooses $\mathcal A\gets \{\frac{iu}{2}\;:\;i\in \mathbb N\}$ as the set of anchors. For each center position, the algorithm initializes $\R[c]\gets \operatorname{NaivePalFind}(S,\,c,\,u_0)$.

In the main loop, the algorithm computes the matching-self-convolution around the anchor positions, as described. When a palindrome is detected at center $c$ by the computation of its closest anchor $a$, the approximation array $\R$ is updated with the new value. At every doubling of the window size, the algorithm removes half of the anchors to keep the running time low. At the end of each iteration, the algorithm increases the window size $u$ by a factor of $(1+\delta)$. The algorithm terminates when the window size reaches $n$, at which point $\R$ is a $(1+\epsilon)$-approximation of all maximal radii.

The algorithm is presented as~\cref{alg:approx}, and an example can be found in~\cref{ssec:example-approx}.
\begin{algorithm}[ht]
\caption{Approximate All Maximal Palindromes}\label{alg:approx}
\begin{algorithmic}[1]
    \REQUIRE String $S$ of length $n$, approximation parameter $0 < \epsilon \leq 0.5$.
    \ENSURE An array $\R$ with $(1+\epsilon)$-approximation of the maximal palindromes of $S$.

    \STATE $\delta\gets \frac{1}{2}\epsilon$
    \STATE $\R\gets (0,\,\dots,\,0)$ \COMMENT{Radius lower bound for each center.}
    \STATE $u', u\gets \operatorname{AlignUpToPowerOfTwo}(\frac{1}{\delta})$
    \FORALL {center position $c$}
        \STATE $\R[c]\gets \operatorname{NaivePalFind}(S,c,u)$ \COMMENT{Limits to at most $u$ comparisons.}
    \ENDFOR

    \STATE $\mathcal{A} \gets \{\frac{iu}{2}\;:\;1\le i\le \frac{2n}{u}\}$
    \WHILE{$u < n/2$}        
        \FOR{$a\in \mathcal{A}$}
            \STATE $\ell\gets \max\{1,\;a-u\};\quad r\gets \min\{n,\;a+u\}$
            \STATE $M\gets \mconv(S[\ell\dd r])$
            \FOR {every center $c$ identified by $M$}
                \STATE $\tilde{r}_c\gets \min(c-\ell,r-c)$\label{alg:lin:pal-len}
                \STATE $\R[c]\gets \max\{\R[c],\tilde{r}_c\}$
            \ENDFOR
        \ENDFOR
        \STATE $u \gets (1+\delta)u$
        \IF {$u\ge 2u'$}
            \STATE $u'\gets 2u'$
            \STATE $\mathcal A\gets \{\mathcal A[2i]\;:\;2i\le |\mathcal A|\}$ \label{alg:lin:rem} \COMMENT{Remove every second element from $\mathcal{A}$.}
        \ENDIF
    \ENDWHILE

    \RETURN $\R$
\end{algorithmic}
\end{algorithm}
\paragraph*{Proof notation.} Throughout the correctness and runtime proofs, we refer to variable snapshots of the algorithm execution. We do so by denoting variable $x$ at iteration $i$ as $x_i$. For example, the variable $u_i$ is the window size in the $i$-th iteration of the outer loop. For this particular example, the relation $u_i=(1+\delta)^{i-1}u_0$ holds where $u_0$ is the initial window size ($\approx \frac{1}{\delta}$). For simplicity, we do not refer to any of the internal loops' variables.

We refer to the factor $S[\max\{1,\,a-u_i\}\dd \min\{n,\,a+u_i\}]$ as the {\em window} of anchor $a$ at iteration $i$.

We define the anchor of center $c$ at iteration $i$ as the anchor that maximizes the radius of the palindrome centered at $c$. Usually, it would be the anchor maximizing $u-|c-a|$. However, next to the borders the radius might be prematurely trimmed.
Formally, the radius computed for center $c$ in iteration $i$ at anchor $a$ is:
\[
\operatorname{rad}(c,\;i,\;a)=\min\{c-\max\{1, a-u_i\},\; \min\{n, a+u_i\}-c\},
\]
and the anchor of center $c$ in iteration $i$ is:
\[
\operatorname{anchor}(c,\,i)\coloneq \argmax_{a\in \mathcal A_i} \operatorname{rad}(c,\;i,\;a).
\]
A center has at most two candidate anchors; when there are two, we choose the one with the smaller index.

\begin{lemma}[Runtime]
    The running time of~\cref{alg:approx} is $\Oh{\frac{1}{\epsilon}n\log^2 n}$.
\end{lemma}

\begin{proof}
In the main (outer) loop, we start with a window size $u_0\approx \frac{1}{2\epsilon}$ and multiply it by $(1+\delta)$ at each step, until it reaches $n$, resulting in $\Oh{\log_{1+\delta} n}=\Oh{\frac{1}{\epsilon}\log n}$ iterations. 

At iteration $i$, we perform a matching convolution and iterate its results on strings of size $\Oh{u_i}$. We do so for each anchor, resulting in $\Oh{|\mathcal A_i|u_i\log u_i}$ time for iteration $i$.

The algorithm initially starts with $\frac{2n}{u_0}$ anchors. Every time the window size is doubled, every second anchor is deleted. Therefore, there are at most $\frac{4n}{u_i}= \Oh{\frac{n}{u_i}}$ anchors in $\mathcal A_i$, hence the running time of iteration $i$ is $\Oh{n\log u_i}\subseteq \Oh{n\log n}$.

The total running time is thus $\Oh{\frac{1}{\epsilon}\log n \cdot n\log n}=\Oh{\frac{1}{\epsilon}n\log^2 n}$, as required.
\end{proof}

\begin{lemma}[Correctness]
    The result of~\cref{alg:approx} is a $(1+\epsilon)$-approximation for every center.
\end{lemma}

\begin{proof}
Our proof demonstrates that the computed radius estimate for any center grows by a factor of at most $(1+\epsilon)$ per iteration, which implies that the approximation will be in range. We denote by $j$ the last iteration where $c$ is identified as a palindrome \emph{by its anchor}. 

In each iteration $i$, the anchors are spaced by $\le \frac{u_i}{2}$ indices. Denote the two anchors closest to $c$ in iteration $j$ as $a_0,\;a_1$, where $a_0\le c\le a_1$. Since both satisfy $|c-a_0|,\;|c-a_1|\le \frac{u_j}{2}$, at least one of them---say $a_0$---lies between $c$ and the string center $\frac{n}{2}$ (or coincides with $c$). The window of $a_0$ therefore extends at least as far as a window centered at $c$ would, so $\operatorname{rad}(c,\; j,\;a_0)\ge u_j - |c-a_0| \ge \frac{u_j}{2}$. Therefore, $\R[c]\ge \frac{u_j}{2}$.

Denote $a=\operatorname{anchor}(c,\;j+1)$. The window size $u$ grows by a factor of $(1+\delta)$ every iteration, hence $\operatorname{rad}(c,\;j+1,\;a)\le \operatorname{rad}(c,\;j,\;a)+\delta u_j$. Since $c$ was not identified in the $(j+1)$-th iteration, it follows that $r_c<\operatorname{rad}(c,\;j,\;a)+\delta u_j$, and since $c$ was identified by its anchor in the $j$-th iteration, then $\operatorname{rad}(c,\;j,\;a)\le \R[c]$.
We evaluate the ratio $\frac{r_c}{\R[c]}$:
\[
\frac{r_c}{\R[c]}
<
\frac{\R[c]+\frac{1}{2}\epsilon u_j}{\R[c]}
=
\frac{2\R[c]+\epsilon u_j}{2\R[c]}
=
1+\frac{\epsilon u_j}{2\R[c]}
\le
1+\frac{\epsilon u_j}{2\cdot\frac{1}{2}u_j}
=
1+\frac{\epsilon u_j}{u_j}
=
1+\epsilon,
\]
which is the desired approximation bound.

\end{proof}

%% file: figures/example_approx.tex
\newcommand{\stringtemplatex}[6]{
    \def\c{5};
  \foreach \v [count=\i] in {b, $\phi$, b, a, a, c, $\phi$, c, e, c, $\phi$, c, c, a, b, $\phi$} {
    \ifthenelse{ \i < #5 \OR \i > #6} {
        \draw[thick] (\i,0) rectangle (\i+1,0.5);
    } {
        \draw[thick, fill=black!15!white] (\i,0) rectangle (\i+1,0.5);
    }
    \node (v\i) at (\i+0.5,0.25) {\v};
    \node at (\i+0.5,-0.25) {\scriptsize \i};
  }

    \pgfmathtruncatemacro{\jump}{floor(16/#1)}
    \pgfmathtruncatemacro{\low}{floor(\jump/2)}
    \foreach \k [evaluate=\k as \i using int(\k*\jump)] in {1,...,#1} {
      \node[circle, draw=red] at ({\i + 0.5}, 0.5) {};
      \node[circle, fill=red, inner sep=1pt] at ({\i + 0.5}, 0.5) {};
      \coordinate (p\i) at (\i+0.5, 0.5);
    }

    \foreach \k [evaluate=\k as \i using int(\k*\jump-\low)] in {1,...,#1} {
      \node[circle, draw=red] at ({\i + 0.5}, 0) {};
      \node[circle, fill=red, inner sep=1pt] at ({\i + 0.5}, 0) {};
      \coordinate (p\i) at (\i+0.5, 0);
    }

  \coordinate (anc) at (\c + 0.5, -0.75);
  \coordinate (arrow) at (\c + 0.5, -0.5);
  \node at (anc) {$\operatorname{anchor}(\c)=#4$};
  \draw[blue, thick, ->] (arrow) -- (p#4.south);

    \draw[blue,thick,decorate,decoration={brace,raise=1em,amplitude=3pt}]
  ($ (v8) + (-#3, 0) $)
  --node[midway, above=1em+5pt] {$u=#2$} 
  ($ (v8) + (#2, 0) $);
}

\begin{figure}[ht]
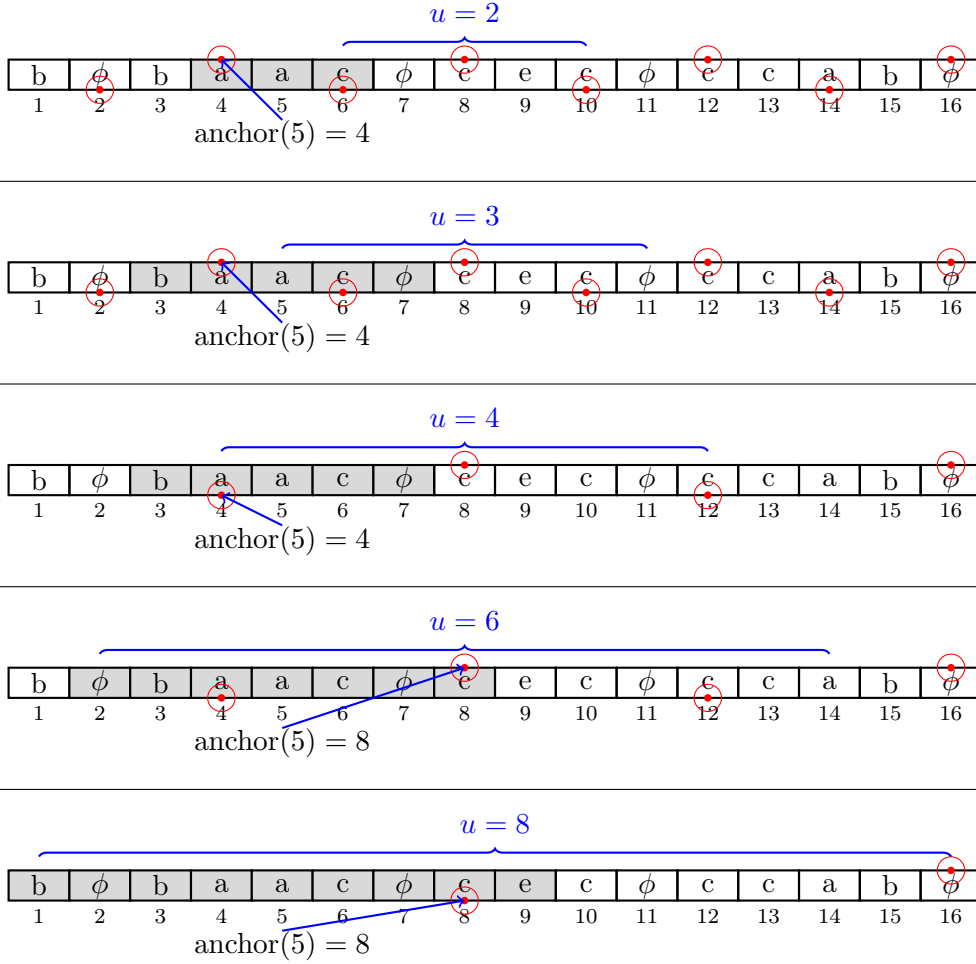

    \newsubfigure[0.8]{
        \stringtemplatex{4}{2}{2}{4}{4}{6}
    }
    \subfigsep
    \newsubfigure[0.8]{
        \stringtemplatex{4}{3}{3}{4}{3}{7}
    }
    \subfigsep
    \newsubfigure[0.8]{
        \stringtemplatex{2}{4}{4}{4}{3}{7}
    }
    \subfigsep
    \newsubfigure[0.8]{
        \stringtemplatex{2}{6}{6}{8}{2}{8}
    }
    \subfigsep
    \newsubfigure[0.8]{
        \stringtemplatex{1}{8}{7}{8}{1}{9}
    }
    \caption{Demonstration of the approximation algorithm. At every step, we perform for each anchor (marked with red) the palindrome prefix/suffix search routine on a string of size $2u+1$, centered around the anchor. In the example, we show the window of anchor $8$, and the anchor of center $5$. The anchor removal procedure occurs every time the window size doubles. The substring centered at $c=5$ that is checked at each iteration is highlighted in gray.}
    \label{fig:approx}
\end{figure}

%% file: results/with_mismatches.tex
In this section, we extend all of this paper's algorithms to find all maximal $k$-palindromes, i.e., maximal palindromes whose palindromic mismatch count (as defined in~\cref{s:pre}) is at most $k$.

\subsection{Baseline Algorithm}

The \textbf{LCE algorithm} requires no modification, and runs in time $\Oh{n(k+G)}$ by using the kangaroo method.
The \textbf{LCEW-based algorithm} can also be adapted by using kangaroo, causing each LCEW query to take $\Oh{kT}$ time, and the time to find all $k$-palindromes is $\Oh{n(kT+\frac{G}{T}\log n)}$. We balance $G\log n/T \approx kT$, which results in $T \approx \sqrt{G\log n/k}$.
\begin{itemize}
\item \textbf{For $G = \Omega(k\log n)$}, the LCEW method outperforms the LCE method in runtime.
\item \textbf{Otherwise,} the LCE method outperforms the LCEW method in both runtime and space.
\end{itemize}

\subsection{Convolution-Based Algorithms}

The most significant results are achieved by adapting our convolution-based methods. This extension relies on computing the Hamming distance between prefixes, which can be done via counting convolution ($\cconv$). We denote the time to compute $\cconv$ on a string $S$ of length $n$ as $f_S(n)$. For constant-sized alphabets, the result is immediate.
\begin{observation}
    Let $S$ be a string with wildcards over a \textbf{constant-size alphabet}. The FFT-based algorithms (\cref{alg:precise}, \cref{alg:approx}) can find all maximal $k$-palindromes \textbf{with no asymptotic increase} in runtime, as $f_S(n)= \Oh{n \log n}$.
\end{observation}

For general alphabets, $f_S(n)= \Oh{n\min(\sqrt{n\log n},\,|\Sigma|\log n)}$. We analyze our two algorithms in this context.

\textbf{Approximate Algorithm}
The algorithm is modified by replacing the matching convolution ($\mconv$) with a counting convolution ($\cconv$). If the Hamming distance is $<2k$, we denote it as a $k$-palindrome. We allow $2k$ mismatches since each mismatch appears in both arms of the palindrome. The overall runtime is similar to the original algorithm, with the $\Oh{n\log n}$ factor being replaced by $f_S(n)$, resulting in
$\Oh{\left(\frac{1}{\epsilon}\right) f_S(n)\log n}$
(Note: If $f_S(n) = \Omega(n^{1+\alpha})$ for some positive constant $\alpha$, the $\log n$ factor is absorbed.).

\textbf{Precise Algorithm}
Similar to before, we rebalance the runtime of phase one and phase two in~\cref{alg:precise}. The runtime is $T(u)= \Oh{\frac{n}{u} \cdot f_S(n) + nu}$, where $u$ is the block size:

\[
\frac{n}{u} f_S(n) \approx nu \implies u^2 \approx f_S(n) \implies u = \sqrt{f_S(n)}.
\]

Substituting $u$ into $T(n)$ results in $T(n)= \Oh{n\sqrt{f_S(n)}}$. \\
Given the upper bound $f_S(n)= \Oh{n\sqrt{n\log n}}$, we achieve a linear-space algorithm with a total runtime of 
$\Oh{n\sqrt{n\sqrt{n \log n}}} = \Oh{n^{7/4}\log^{1/4}{n}}$, significantly improving the $\Oh{n^2}$-time naive algorithm.

\subsection{No wildcards}

When the input text contains no wildcards (i.e., the task is to find palindromes with $k$ mismatches only), our previous methods still work, but their running time can be improved. While the work by Amir et al.~\cite{alp:04} on $k$ mismatches seems to be of aid, the original work was designed for traditional pattern matching, while we also need to detect all prefix and suffix alignments. Since the original paper uses the ``mark-and-sweep'' technique internally, a non-trivial adaptation of their technique is required.

A later result presented by Kaplan et al.~\cite{kps:06} gives a stronger (and slower) oracle that identifies the position of the $k$-th mismatch between a pattern and a text. Although their result does not address prefix/suffix alignments, the adaptation is natural. Our goal is to use this result to perform the {\em $k$-convolution}, defined below:
\begin{definition}[$k$-convolution]
    Let $S,\,T$ be two solid strings. The {\em $k$-convolution} of $S,\,T$, denoted by $\kconv(S,\,T)$ is the array $C_k$ (defined by the counting convolution):
    \[
    C_k[i]\coloneq \min(C[i],\,k+1),\qquad C=\cconv(S,\,T)
    \]
\end{definition}

The motivation for the $k$-convolution follows directly from our algorithms, which ignore any matches with more than $k$ unique mismatches (i.e., Hamming distance of $2k$). We proceed to state the results:
\begin{lemma}\label{lem:kconv}
    Let $S,\,T$ be two solid strings of the same length $n$. The $k$-convolution of $S,\,T$ can be computed in time $\Oh{nk^{2/3}\log^{1/3}{n} \log k}$.
\end{lemma}

The proof idea is to append sentinel values to the strings such that suffixes are considered part of the original string. Since the sentinel values are guaranteed to create a mismatch, the Hamming distance can be derived from them. We proceed to the formal proof:
\begin{proof}[Proof for~\cref{lem:kconv}]
    We refer to the oracle that finds the position of the $k$-th mismatch as $\operatorname{Pos}_k(S,\,T)$, assuming $|T|\le |S|$.

    Let $\$$ be a fresh symbol not in the alphabet. Denote $S'=S\$^n$ (i.e., the string $S$ with $n$ sentinel characters appended), and invoke $\operatorname{Pos}_{k+1}(S',\,T)$. Let $j_i$ be the position of the $(k+1)$-th mismatch between $T$ and $S'[i\dd i+n-1]$. If $j_i\le n-i+1$, then the $(k+1)$-th mismatch is within the borders of the original string $S$, and the $k$-convolution should be $k+1$ for position $i$. However, if $j_i>n-i+1$, it implies that exactly $j_i+i-1-n$ mismatches were counted by comparing characters to the new $\$$ symbols. Therefore, the following computation satisfies the $k$-convolution (for prefixes of $T$):
    \[
    \kconv(S,\,T)[i+n-1]=\begin{cases}
        k+1 & j_i\le n-i+1 \\
        n+k+2-j_i-i & \text{otherwise}
    \end{cases}
    \]

    The computation of the first $n$ values of $\kconv$ is symmetrical and achieved by padding $T$ instead of $S$. The original algorithm does not report the position $j_i$ if there are fewer than $k+1$ mismatches at that alignment. However, by padding $S$ and $T$ by additional $k+1$ distinct sentinels, we are guaranteed that every alignment will incur at least $k+1$ mismatches.
\end{proof}

Although the $k$-convolution can speed up individual convolutions, we show that it is never optimal when integrated into the precise algorithm, as for every parameter range for $k$, either Abrahamson's method or the kangaroo jump technique outperform it.

The precise algorithm runtime is $\Oh{n\sqrt{f_S(n)}}$. Therefore, for $k= \Oh{\sqrt{f_S(n)}}$, the LCE method outperforms it. Since $k$ is polynomial in $n$, we have $f_S(n)= \Oh{nk^{2/3}\log^{4/3}{n}}$, and the LCE algorithm outperforms it for:
\[
k\le \sqrt{nk^{2/3}\log^{4/3}{n}}\to k^{2/3}\le n^{1/2}\log^{2/3}{n}\to k\le n^{3/4}\log n.
\]

However, for $k = \omega(n^{3/4}\log n)$, the worst-case running time of the $k$-convolution is:
\[
\Theta(nk^{2/3}\log^{4/3}{n})=\omega( n\sqrt{ n\log n}),
\]
hence it is outperformed by Abrahamson's method.

%% file: future_work.tex
Our current approach leaves room for improvement in terms of running time. Specifically, we leave open the problem of designing an algorithm running in time $\Oh{n^{1.5-\alpha}}$ for some $\alpha>0$ that finds all maximal palindromes in a string with wildcards. Moreover, an algorithm for finding the longest palindrome in the presence of wildcards and mismatches is yet to be described.